\documentclass[12pt]{article}

\usepackage{fullpage,float,amssymb,amsmath,amsthm}
\usepackage{epsf}

\usepackage[usenames]{color}

\usepackage[colorlinks=true,
linkcolor=webgreen,
filecolor=webbrown,
citecolor=webgreen]{hyperref}

\definecolor{webgreen}{rgb}{0,.5,0}
\definecolor{webbrown}{rgb}{.6,0,0}

\usepackage{color}
\usepackage{url}

\author{Jean-Paul Allouche \\
CNRS, IMJ-PRG, UPMC \\
UPMC, Case 247, 4 Place Jussieu \\
F-75252 Paris Cedex 5 \\
France \\
{\tt jean-paul.allouche@imj-prg.fr} \\
\and
Julien Cassaigne \\
CNRS, Institut de Math. de Marseille \\
Campus de Luminy \\
F-13288 Marseille Cedex 9 \\
France \\
{\tt julien.cassaigne@math.cnrs.fr}
\and
Jeffrey Shallit \\
School of Computer Science \\
University of Waterloo \\
Waterloo, ON  N2L 3G1 \\
Canada \\
{\tt shallit@cs.uwaterloo.ca} 
\and 
Luca Q. Zamboni \\
Institut Camille Jordan \\
Universit\'e Claude Bernard, Lyon 1 \\
43 Boulevard \ du 11 novembre 1918 \\
F-69622 Villeurbanne Cedex \\
France \\
{\tt lupastis@gmail.com}
}

\title{A Taxonomy of Morphic Sequences}

\begin{document}

\maketitle

\theoremstyle{plain}
\newtheorem{theorem}{Theorem}
\newtheorem{corollary}[theorem]{Corollary}
\newtheorem{lemma}[theorem]{Lemma}
\newtheorem{proposition}[theorem]{Proposition}

\theoremstyle{definition}
\newtheorem{definition}[theorem]{Definition}
\newtheorem{example}[theorem]{Example}
\newtheorem{conjecture}[theorem]{Conjecture}

\theoremstyle{remark}
\newtheorem{remark}[theorem]{Remark}

\def\Que{{\mathbb{Q}}}

\begin{abstract}
In this note we classify sequences according to whether they
are morphic, pure morphic, uniform morphic, pure uniform morphic,
primitive morphic, or pure primitive morphic, and for each possibility
we either give an example or prove that no example is possible.
\end{abstract}

\section{Introduction}
Sequences (also known as right-infinite words; we use the terms
interchangeably) that arise from the iteration of a morphism appear in many
different places in the mathematical and theoretical computer science
literature.  To name just a few examples, they play an important role in the
theory of words avoiding patterns \cite{Thue:1906,Thue:1912,Berstel:1992},
in the theory of Sturmian sequences \cite{Berstel&Seebold:2002}, in
L-systems used in computer graphics \cite{Prusinkiewicz&Lindenmayer:1990},
and in the theory
of algebraic series in positive characteristic
\cite{Christol:1979,Christol&Kamae&MendesFrance&Rauzy:1980}.

Different kinds of morphisms give rise to different kinds of sequences,
with different properties.
In this paper we consider some of these variations and properties and
classify sequences according to whether they do, or do not, satisfy
these properties.    The goal is to create a relatively complete
taxonomy of the different kinds of behavior exhibited my morphic
sequences, as well as to illustrate the many different techniques that
can be used to show that a sequence exhibits, or does not
exhibit, a certain property.

Let us fix our notation.  Let $\Sigma^*$ denote the set of all finite words
over the alphabet $\Sigma$, including the empty word $\epsilon$, and let
$\Sigma^\omega$ denote the set of all right-infinite words over $\Sigma$.
We write $\Sigma^\infty = \Sigma^* \ \cup \ \Sigma^\omega$.
Let $h:\Sigma^* \rightarrow \Sigma^*$ be a
morphism, that is, a map obeying $h(xy) = h(x) h(y)$ for all words
$x, y \in \Sigma^*$.  From the definition, it suffices to define
$h$ on each element of $\Sigma$.  A morphism $h:\Sigma \rightarrow \Delta$
is said to be {\it non-erasing} if $h(a) \not= \epsilon$ for all
$a \in \Sigma$.  A letter $a$ is said to be {\it growing\/} for $h$ if
$\lim_{n \rightarrow \infty} |h^n (a)| = \infty$; otherwise it is
{\it non-growing}.

We can define $h$ on
$\Sigma^\omega$ in the obvious way:  $h(a_1 a_2 \cdots) = 
h(a_1) h(a_2) \cdots$.

We can iterate $h$, writing $h^2$ for the
composition $h \circ h$, $h^3$ for $h \circ h \circ h$, etc.
If there exists a letter $a \in \Sigma$ and
$x \in \Sigma^*$ such that
\begin{enumerate}
\item $h(a) = ax$; and
\item $h^i(x) \not= \epsilon$ for all $i \geq 0$,
\end{enumerate}
then $h$ is said to be {\it prolongable on $a$}.
In this case,
iterating $h$ on $a$ produces a sequence of words of increasing 
length,
$$a, h(a), h^2(a), \ldots$$
where each word is a proper prefix of the word that follows.
In the limit, this tends to the infinite word
$$h^\omega(a) := a \, x \, h(x) \, h^2(x) \, h^3(x) \cdots \in \Sigma^\omega , $$
which is a fixed point of $h$; that is,
$$ h(h^\omega(a)) = h^\omega(a).$$
We say that the infinite word $h^\omega(a)$ is {\it generated} by $h$.

If $w$ is a finite word, then $|w|$ denotes the length of $w$.
If $w$ is a finite or infinite word, then $w[n]$ denotes the
$n$'th symbol of $w$, and $w[i..j]$ represents the factor
of $\bf w$ beginning at position $i$ and ending at position $j$.

If $x = uvw$, then we say that $u$ is a {\it prefix\/} of $x$, that
$w$ is a {\it suffix\/} of $x$, and $v$ is a {\it factor\/} of $x$.

Let $\Sigma_k = \{ 0,1,\ldots, k-1 \}$.  If $w \in \Sigma_k^*$, we
let $[w]_k$ denote the integer represented by $w$ in base $k$.

We start by discussing different kinds of morphisms.

\subsection{Pure morphic words}

If $\bf w$ is an infinite word over $\Sigma$ and
there exists a morphism $h$, prolongable on $a$, such that
${\bf w} = h^\omega (a)$, then $\bf w$ is said to be {\it pure morphic}.

\begin{example}
One of the most famous pure morphic words is the {\it Fibonacci word}
$${\bf f} =  01001010 \cdots,$$
which is generated by the morphism $\varphi$ defined by
$$0 \rightarrow 01; \quad\quad\quad 1 \rightarrow 0.$$
See, e.g., \cite{Berstel:1980,Berstel:1986}.
\label{fib}
\end{example}

\subsection{Pure uniform morphic words}

A morphism $h$ is said to be {\it $k$-uniform}
if $|h(a)| = k$ for all $a \in \Sigma$.
It is said to be {\it uniform} if it is uniform for some $k \geq 2$.
If an infinite word is generated by a uniform prolongable morphism,
then it is said to be {\it pure uniform morphic}.  

\begin{example}
One of the most
famous pure uniform morphic words is the {\it Thue-Morse word}
$$ {\bf t} = 01101001 \cdots,$$
which is generated by the morphism $\mu$ defined by
$$0 \rightarrow 01; \quad\quad\quad 1 \rightarrow 10.$$
See, for example, \cite{Thue:1906,Thue:1912,Berstel:1995,Allouche&Shallit:1999}.
\end{example}

\subsection{Morphic words}

We can also apply a coding (a $1$-uniform morphism from $\Sigma$ to
a possibly different alphabet $\Delta$) to a morphic word.  This has
the effect of renaming the symbols.  One way to think about this is
to give letters in $\Sigma$ subscripts, and then the effect of $\tau$
is to erase the subscripts.   We use the shorthand
$a_1 a_2 \cdots a_n \rightarrow b_1 b_2 \cdots b_n$ to represent
the coding $a_i \rightarrow b_i$ for $i = 1, 2, \ldots n$.
If an infinite word is expressible as the coding of a pure morphic word,
it is said to be ${\it morphic}$.  

\begin{example}
The Fibonacci numbers are defined by $F_0 = 0$, $F_1 = 1$, and
$F_n = F_{n-1} + F_{n-2}$.  The {\it characteristic word} $\chi_F 
= (\chi (n))_{n \geq 0}$ of
the Fibonacci numbers is defined by
\begin{equation}
\chi (n) = \begin{cases} 
	1, & \text{ if $n = F_i$ for some $i \geq 0$;}\\
	0, & \text{ otherwise.}
	\end{cases}
\end{equation}
Thus 
$$\chi = 11110100100001 \cdots .$$
This word is morphic, as it is generated by the morphism sending
\begin{eqnarray*}
c_0 & \rightarrow & c_0 e_0 \\
e_0 & \rightarrow & c_1 \\
c_1 & \rightarrow & c_1 e_1 \\
e_1 & \rightarrow & c_2 \\
c_2 & \rightarrow & c_2 e_1 
\end{eqnarray*}
followed by the coding $c_0 e_0 c_1 e_1 c_2 \rightarrow 11100$.
See, for example, \cite[Example 7, p.~168]{Cobham:1972}.
\end{example}

\subsection{Uniform morphic words}

If an infinite word is expressible as the coding of
a pure uniform morphic word, it is said to be {\it uniform morphic}.
From Cobham's theorem \cite{Cobham:1972}, we know that an infinite word
is uniform morphic if and only if it is automatic.  (A $k$-automatic
word is one where the $n$'th term is computed by a finite automaton
reading the base-$k$ expansion of $n$ as input and producing an output
associated with the last state reached; a word is said to be automatic
if it is $k$-automatic for some integer $k \geq 2$.)

\begin{example}
Consider the uniform morphic word $\bf s$ generated by the morphism
\begin{eqnarray*}
a & \rightarrow & ab \\
b & \rightarrow & ac \\
c & \rightarrow & db \\
d & \rightarrow & dc .
\end{eqnarray*}
Apply the coding sending $abcd \rightarrow 0011$ to $\bf s$;  the result is
$$ {\bf r} = 0001001000011101 \cdots,$$
the {\it Golay-Rudin-Shapiro} sequence.  (See, for example,
\cite{Golay:1949,Golay:1951, Shapiro:1951,Rudin:1959}.)
It is not hard to see that
${\bf r}[n]$ equals the parity of the number of (possibly overlapping)
occurrences of $11$ in the binary expansion of $n$.
\label{rudin}
\end{example}

\subsection{Pure primitive morphic words}

A morphism $h:\Sigma^* \rightarrow \Sigma^*$
is called {\it primitive} if there exists an integer $n \geq 1$
such that for all $a \in \Sigma$, the word $h^n(a)$ contains at least
one occurrence of each symbol in $\Sigma$.  For example,
the Fibonacci morphism $\varphi$ is primitive, since
$\varphi^2 (0) = 010$ and $\varphi^2 (1) = 01$.
An infinite word $\bf w$ is {\it pure primitive morphic} if it is the fixed
point of some primitive morphism.

\begin{example}
Consider the morphism $\gamma$ defined by
\begin{align*}
\gamma(0) &= 03 \\
\gamma(1) &= 43 \\
\gamma(3) &= 1 \\
\gamma(4) &= 01.
\end{align*}
Then 
\begin{align*}
\gamma^6(0) &= 03143011034343031011011 \\
\gamma^6(1) &= 03143011031011011031011011 \\
\gamma^6(3) &= 03143034343034343 \\
\gamma^6(4) &= 0314301103434303143034343034343 , 
\end{align*}
so $\gamma$ is primitive.
The infinite word $0314301103434303101101103\cdots$ generated
by $\gamma$ is thus pure primitive morphic.  It was studied
in \cite{Cassaigne&Currie&Schaeffer&Shallit:2014} and has the interesting
property that it avoids additive cubes, that is, three consecutive blocks
of the same length and same sum.
\end{example}

\subsection{Primitive morphic words}

An infinite word is
{\it primitive morphic} if it is
the image, under a coding, of some primitive morphism.

\begin{example}
An interesting example of a primitive morphic word appears in
\cite{Du&Mousavi&Rowland&Schaeffer&Shallit:2017}.  Consider the
morphism $g$ and coding $\rho$ defined as follows:
\begin{align*}
g(a) &= abcab  \quad & \rho(a) = 0 \\
g(b) &= cda \quad & \rho(b) = 0 \\
g(c) &= cdacd \quad & \rho(c) = 1 \\
g(d) &= abc \quad &  \rho(d) = 1
\end{align*}
and form the infinite word ${\bf R} = \rho(g^\omega(a))$.  By considering $g^2$
we see that $g$ is primitive and so $\bf R$ is primitive morphic.
The word $\bf R$, sometimes called the Rote-Fibonacci word,
is an aperiodic word that avoids the pattern $x x x^R$, where $x^R$ denotes the
reversal of $x$.
\end{example}

\subsection{Uniform primitive morphic words}

Finally, an infinite word is {\it pure uniform primitive morphic} if it
is the fixed point of a primitive morphism that is also uniform, and
{\it uniform primitive morphic} if it is the image, under a coding of
such a word.

\subsection{Recurrence and uniform recurrence}

An infinite word $\bf w$ is said to be {\it recurrent} if every finite
factor $x$ appearing in $\bf w$ appears infinitely often in $\bf w$.
If, in addition, for each factor $x$ there is a constant $c(x)$ such
that every two consecutive occurrences of $x$ in $\bf w$ are separated by
no more than $c(x)$ symbols, then we say that $\bf w$ is
{\it uniformly recurrent}.
Synonyms for uniformly recurrent in the literature include
``almost periodic'' \cite{Cobham:1972} and ``minimal'' \cite{Fogg:2002}.

A basic result connecting uniform recurrence with other properties
is the following theorem of Cobham
\cite[Theorem 5, p.\ 178]{Cobham:1972}:

\begin{theorem}
\leavevmode
\begin{itemize}
\item[(a)] If $\bf w$ is a primitive morphic word, then it is
uniformly recurrent.  
\item[(b)] If $\bf w$ is uniformly recurrent and uniform morphic,
then it is uniform primitive morphic.
\end{itemize}
\label{cobham}
\end{theorem}

We now mention some other useful results:

\begin{theorem}
Let $\bf u$ be a fixed point of a morphism $\sigma$ ``without useless
letters" (i.e., the alphabet on which $\sigma$ is defined is exactly
the set of letters occurring in $\bf u$). Furthermore suppose that for
each letter $a$ in $\bf u$ the length of the iterates
$|\sigma^k(a)|$ tends to infinity.
If the sequence $\bf u$ is uniformly recurrent, then $\sigma$ is primitive.
\end{theorem}

\begin{proof}
By the hypothesis on the letters being all useful, 
it follows that
for every letter $a$ that $\sigma^k(a)$ is a factor of $\bf u$ for all 
$k \geq 0$.
Since $u$ is uniformly recurrent, for every letter $b$ there exists a length $\ell_b$
such that every factor of $u$ of length $\geq \ell_b$
contains at least one $b$.
Taking $\ell = \max_{b \in \Sigma} \ell_b$, every factor of 
$\bf u$ of length $\geq \ell$ contains at
least one copy of each letter.
On the other hand, for each $a$ we know that $|\sigma^k(a)|$ is
unbounded.  Hence there exists $k_a$ such that for all $k \geq k_a$
one has $|\sigma^{k_a}(a)| \geq \ell$. Taking $K = \max k_a$,
we see that for each
$k \geq K$ and for each letter $a$, the word
$\sigma^k(a)$ contains at least one copy of
each letter. In other words, $\sigma$ is primitive.
\end{proof}

\begin{corollary}
If a sequence is pure uniform morphic and uniformly recurrent,
then the sequence
is pure uniform primitive morphic.
\end{corollary}

\begin{proof}
Clean out the useless letters.
\end{proof}

\begin{corollary}
If a sequence is pure uniform morphic and primitive morphic, then
the sequence is pure uniform primitive morphic.
\label{allou}
\end{corollary}

\begin{proof}
The sequence is uniformly recurrent. Apply the previous corollary.
\end{proof}

We recall a result from 
Queff{\'e}lec \cite[Prop.~5.5, p.~130]{Queffelec:2010}:

\begin{theorem}
Let $h:\Sigma^* \rightarrow
\Sigma^*$ be a morphism, prolongable on $a$, and
suppose all letters of $\Sigma$ are growing.  Then
$h^\omega(a)$ is uniformly recurrent if and only if
$h$ is primitive.
\label{queffelec}
\end{theorem}

Finally, we also mention a recent important result of
Durand \cite[Thm.~3, p.~124]{Durand:2013}:

\begin{theorem}
Let $\bf u$ be a morphic sequence that is uniformly recurrent.
Then $\bf u$ is primitive morphic.  
\label{durand}
\end{theorem}

\begin{remark}
More precisely, Durand proves in \cite{Durand:2013} that a sequence is uniformly recurrent 
and morphic if and only if it is ``primitive substitutive'',
but his theorem also implies the statement 
above (F. Durand, private communication, June 2017).
\end{remark}

\section{The classification}

From the preceding section, a word can be classified in ten different
ways:

\begin{itemize}
\item[P1:] pure morphic
\item[P2:] morphic
\item[P3:] pure uniform morphic
\item[P4:] uniform morphic
\item[P5:] pure primitive morphic
\item[P6:] primitive morphic
\item[P7:] pure uniform primitive morphic
\item[P8:] uniform primitive morphic
\item[P9:] uniformly recurrent
\item[P10:] recurrent
\end{itemize}

However, these ten properties are clearly not independent.  We have the
following trivial implications:

\begin{itemize}
\item P1 $\implies$ P2
\item P3 $\implies$ P1, P2, P4
\item P4 $\implies$ P2
\item P5 $\implies$ P1, P2, P6
\item P6 $\implies$ P2
\item P7 $\implies$ P1, P2, P3, P4, P5, P6, P8
\item P8 $\implies$ P2, P4, P6
\item P9 $\implies$ P10
\end{itemize}

Theorem~\ref{cobham} (a) tells us that P6 $\implies$ P9, and
Theorem~\ref{cobham} (b) tells us (P6 and P4) $\implies$ P8.
Corollary~\ref{allou} tells us that (P3 and P9) $\implies$ P7.
Theorem~\ref{durand} tells us that (P2 and P9) $\implies$ P6.
All these restrictions lower the total
number of possibilities from 1024 to 20:

\begin{itemize}
\item[(a)] Neither morphic nor recurrent.

\item[(b)] Recurrent, but neither morphic nor uniformly recurrent.

\item[(c)] Uniformly recurrent, but not morphic.

\item[(d)] Morphic; but neither pure morphic, uniform morphic,
primitive morphic, nor recurrent.

\item[(e)] Morphic and recurrent; but neither pure morphic, uniform morphic,
primitive morphic, nor uniformly recurrent.

\item[(f)] Primitive morphic; but neither pure morphic nor uniform morphic.

\item[(g)] Uniform morphic; but neither pure morphic, primitive morphic,
nor recurrent.

\item[(h)] Uniform morphic and recurrent; but neither pure morphic nor
primitive morphic.

\item[(i)] Uniform primitive morphic; but not pure morphic.

\item[(j)] Pure morphic; but neither uniform morphic, primitive morphic,
nor recurrent.

\item[(k)] Pure morphic and recurrent; but neither uniform morphic,
primitive morphic, nor uniformly recurrent.

\item[(l)] Pure morphic and primitive morphic; but neither uniform morphic
nor pure primitive morphic.

\item[(m)] Pure primitive morphic; but not uniform morphic.

\item[(n)] Pure morphic and uniform morphic, but neither pure uniform
morphic, primitive morphic, nor recurrent.

\item[(o)] Pure morphic and uniform morphic and recurrent,
but neither pure uniform morphic nor primitive morphic.

\item[(p)] Pure morphic and uniform primitive morphic;
but neither pure uniform morphic nor pure primitive morphic.

\item[(q)] Pure primitive morphic and uniform primitive morphic; but
not pure uniform morphic.

\item[(r)] Pure uniform morphic; but neither primitive morphic nor recurrent.

\item[(s)] Pure uniform morphic and recurrent; but not primitive morphic.

\item[(t)] Pure uniform primitive morphic.

\end{itemize}

In this note we give examples of all 20 possibilities.  Of course,
examples of some of these cases are very well-known; it is our point to
collect these examples in one place and to illustrate each of the 
20 classes.  Inexplicably, the authors of \cite{Allouche&Shallit:2003} 
failed to do this explicitly.

Before we get to the examples, we recall several more useful results.

The first few involve frequency of letters.
Define the frequency of a symbol $a$ in an infinite word $\bf w$ to 
be the quantity $\lim_{n \rightarrow \infty} {{{{\bf w}[0..n-1]}_a} \over{n}}$,
if it exists.  

\begin{theorem} 
Suppose the frequency $\alpha$ of the letter $a$ in the word $\bf w$ exists.
Then
\begin{itemize}
\item[(a)] If $\bf w$ is morphic then $\alpha$ is algebraic.
\item[(b)] If $\bf w$ is uniform morphic then $\alpha$ is rational.
\end{itemize}
\label{freq}
\end{theorem}

For a proof, see \cite[Theorem 8.4.5, p.~268]{Allouche&Shallit:2003}.

\begin{proposition}
Let $\bf w$ be a morphic sequence and let
$a$ be a letter occurring infinitely
often in $\bf w$.  Then the number of occurrences of
$a$ in a prefix of length $n$ of $\bf w$ is
$\Omega(\log n)$.  
\label{positions}
\end{proposition}

\begin{proof}
This follows from the matrix representation of the morphism.
See, for example, \cite[Cor.~8.2.4, p.~249]{Allouche&Shallit:2003}.
\end{proof}

The last results involve subword complexity (or factor complexity), the
number of distinct factors of length $n$.    We recall the following
theorem of Pansiot \cite{Pansiot:1984a}:

\begin{theorem}
If $\bf w$ is a pure morphic word, then the subword complexity of 
$\bf w$ is in $\Theta(1)$,
$\Theta(n)$, $\Theta (n\log\log n)$, $\Theta(n \log n)$, or $\Theta(n^2)$.
\label{subword}
\end{theorem}

We also recall the following theorem:

\begin{theorem}
If $\bf w$ is a uniform morphic word or a primitive morphic word that
is not ultimately periodic,
then the subword complexity of $\bf w$ is $\Theta(n)$.
\label{uniprim}
\end{theorem}

\begin{proof}
For uniform morphic words, see \cite[Thm.~2, p.~171]{Cobham:1972}.
For primitive morphic words, see \cite{Michel:1975,Michel:1976a,Pansiot:1984a}.
\end{proof}

\section{The examples}

We now turn to providing examples of all of the 20 possibilities listed
in the previous section.

\begin{example}  {\bf (a) A word that is neither morphic nor
recurrent.}  

Consider the binary word
$01100010000000000000000010 \cdots$
that is the characteristic sequence of the factorials
$1, 2, 6, 24, \ldots$.   
This word has
$n = O((\log N)/(\log\log N))$ $1$'s in a prefix of length $N = n! + 1$,
and so by Proposition~\ref{positions}
it cannot be morphic.
It is evidently not recurrent because the factor $11$ appears only once.
\end{example}

\begin{example} {\bf (b) A word that is recurrent,
but neither morphic nor uniformly recurrent.}

Consider the
binary word 
$$ {\bf b} = 11011100101110111 \cdots$$
formed by the concatenation of the binary expansions of
$1, 2, 3, \ldots$ in order.   
This word clearly has $2^n$ distinct factors of length $n$, and hence
by Theorem~\ref{subword}
cannot be morphic.
Nor is $\bf b$ uniformly recurrent, because it contains arbitrarily long
blocks of $0$'s.  However, it is recurrent.
\end{example}

\begin{example}  {\bf (c) A word that is uniformly recurrent,
but not morphic.}

Consider the
Sturmian characteristic word
${\bf s}_{\alpha} = s_0 s_1 s_2 \cdots$ defined by
$s_n = \lfloor (n+1) \alpha \rfloor - \lfloor n\alpha \rfloor$
for $n \geq 0$.  It is well-known that all such words are uniformly
recurrent, see, e.g., \cite[Proposition 3.17, p.~186]{Rigo:2014}.

However, for ${\bf s}_{\alpha}$ to be morphic, the number $\alpha$ must be
a quadratic irrational \cite[Prop.~2.11]{Berthe&Holton&Zamboni:2006}.
So take $\alpha = \pi$, for example.
\end{example}

\begin{example}
{\bf (d) A word that is morphic; but neither pure morphic, uniform morphic,
primitive morphic, nor recurrent.}

Take the Fibonacci word mentioned above in Example~\ref{fib}, and
change the first two symbols to $2$, giving the word
$$ {\bf f}' = 22001010 \cdots .$$
Now ${\bf f}'$ cannot be pure morphic, because if it were, then we would
be able to write it as $h^\omega(2)$ where $h(2)$ begins with $22$.
Then iterating $h$ would produce infinitely many $2$'s, a contradiction.
This word cannot be uniform morphic by 
Theorem~\ref{freq}, because the frequency of the
symbol $0$ is the same as that in the Fibonacci word, namely
$(\sqrt{5} -1)/2$, which would contradict Theorem~\ref{freq}.  Finally,
${\bf f}'$ cannot be primitive morphic by Theorem~\ref{cobham} because 
the symbol 2 only occurs twice, and so ${\bf f}'$ is not recurrent.

However, the word ${\bf f}'$ is morphic.  It is generated by
the morphism $a \rightarrow ab$, $b \rightarrow c$, 
$c \rightarrow cd$, $d \rightarrow c$, followed by the coding
$abcd \rightarrow 2201$.
\end{example}

\begin{example}
{\bf (e) A word that is 
morphic and recurrent; but neither pure morphic, uniform morphic,
primitive morphic, nor uniformly recurrent.}

Consider the fixed point $\bf x$ of the morphism
$a\rightarrow ababb$, $b \rightarrow bc$, $c \rightarrow c$
coded by $\tau$ sending $a\rightarrow 0$ and $b,c\rightarrow 1$.
An easy induction gives
$$ {\bf x} = 0101111010111111111010111101 \cdots =
\prod_{n \geq 1} 0 1^{a(n)} $$
where $a(n) = (\nu_2 (n)+1)^2$ and
$\nu_2 (n)$ is the exponent of the highest power of $2$ dividing $n$.

Then $0 1^n 0 $ occurs in $\bf x$ 
if and only if $n$ is a perfect square.

Following the construction in \cite{Deviatov:2008}, it can be shown that
$\bf x$ has subword complexity $\Theta(n \sqrt{n})$,
so by Theorem~\ref{uniprim} $\bf x$ is not uniform morphic or
primitive morphic.  By Theorem~\ref{subword} 
it is not pure morphic.   It is recurrent but not uniformly recurrent.
\end{example}

\begin{example}
{\bf (f) A word that is
primitive morphic, but neither pure morphic nor uniform morphic.}

From a theorem of Yasutomi \cite{Yasutomi:1997,Berthe&Ei&Ito&Rao:2007},
we know that the Sturmian word
${\bf s}_{\alpha, \rho} = s_0 s_1 s_2 \cdots$ defined by
$s_n = \lfloor (n+1) \alpha + \rho \rfloor - \lfloor n\alpha + \rho \rfloor$
is pure morphic if and only if $\alpha$ is a quadratic
irrational, $\rho \in \Que(\alpha)$, and either
$\alpha' > 1$, $1 - \alpha' \leq \rho' \leq \alpha'$ or
$\alpha' < 0$, $\alpha' \leq \rho' \leq 1- \alpha'$
where $\alpha'$ is the (algebraic) conjugate of $\alpha$.

Now consider the case where $\alpha = (3-\sqrt{5})/2$.  Then 
$\alpha' = (3+\sqrt{5})/2 > 1$.   Take $\rho = 2\alpha = 3-\sqrt{5}$.
Then $\rho' = 3+\sqrt{5} > \alpha'$, so ${\bf s}_{\alpha, \rho}$ is not
pure morphic.  However, this word is just the shift of the Fibonacci
word $\bf f$, and hence is morphic.  In fact, it is easy to see
that ${\bf s}_{\alpha, \rho}  = \tau(h^\omega(a))$, where
$h: a\rightarrow ac, c \rightarrow b, b \rightarrow ac$ and
$\tau(abc) = 100$.  Then $h$ is primitive, as $h^3$ applied to
each letter contains
every letter.  So ${\bf s}_{\alpha, \rho}$ is primitive morphic.
It cannot be uniform morphic by Theorem~\ref{freq}.
\end{example}

\begin{example}
{\bf (g) A word that is uniform morphic;
but neither pure morphic, primitive morphic, nor recurrent.}

As is well known, the morphism $g$ defined by
$2 \rightarrow 210, 1 \rightarrow 20, 0 \rightarrow 1$
generates a squarefree word $g^\omega (2) = 210201\cdots$;
see \cite{Berstel:1979a}.  
Now consider the morphism $h$ defined by
\begin{eqnarray*}
a & \rightarrow & ab \\
b & \rightarrow & ca \\
c & \rightarrow & cd \\
d & \rightarrow & ac 
\end{eqnarray*}
and the coding $\tau$ defined by $\tau(abcd) = 2101$.
Then it is known that $\tau(h^\omega(a)) = g^\omega(2)$; see
\cite{Berstel:1979a}.  Hence $g^\omega(2)$ is $2$-automatic and
hence uniform morphic.

Now take the word $g^\omega(2)$
and change the first $1$ to $2$, obtaining the word
${\bf w} = 220201210120\cdots$.
The resulting word $\bf w$
is $2$-automatic since $g^\omega(2)$ is.
(In fact, it is generated by iterating the morphism
\begin{align*}
0 & \rightarrow 01 & \quad & 1 \rightarrow 23 \\
2 & \rightarrow 24 & \quad & 3 \rightarrow 35 \\
4 & \rightarrow 32 & \quad & 5 \rightarrow 23 
\end{align*}
followed by the coding
$012345 \rightarrow 220211$.)

But $\bf w$ is not pure morphic.  Suppose ${\bf w} = \xi^\omega (2)$.
Then $\xi$ maps $2$ to a word beginning with $22$, 
which means that $22$ occurs infinitely often in $\bf w$,
a contradiction.

Finally, $\bf w$ is neither primitive morphic  nor recurrent
because $22$ only occurs once.
\label{example15}
\end{example}

\begin{example}
{\bf (h) A word that is uniform morphic and recurrent;
but neither pure morphic nor primitive morphic.}

The idea is to construct a word that contains a nonzero
but finite number of $k$-th powers. Such a word cannot be pure uniform
morphic, and it can only be pure morphic if the $k$-th powers are
made up of non-growing letters.

Consider ${\bf v}$,
the fixed point of the morphism $a \rightarrow abba$, $b\rightarrow bccb$,
$c \rightarrow cbbc$.  Apply  the coding
$abc \rightarrow 001$ to get
$$ {\bf x} = 0000011001100000011010011001011001101001100101100000 \cdots .$$

Clearly $\bf x$ is uniform morphic and recurrent.  However, 
$000$ occurs with unbounded gaps between occurrences,
since arbitrarily long factors of Thue-Morse occur, so $\bf x$ is
not uniformly recurrent.

If $\bf x$ were pure morphic, it would
be fixed by some morphism $f$ prolongable on $0$,
then it would start with $f(0)^5$, since $f(0)$ starts with $000001$.
But this is not possible, because there is no factor $000001u000001u0000$.
To see this, note that $000001$ occurs only at positions that are a
multiple of 4, and once $000001u$
is synchronized modulo 4, it can be factored into $\{0000, 0110, 1001 \}$
and decoded.  An occurrence of
$000001u000001u0000$ then
corresponds to an overlap in $\bf v$, which is
impossible since $\bf v$ is overlap-free (applying the coding
$abc \rightarrow 010$
gives the Thue-Morse sequence).

\bigskip

An alternative construction creates a word with
unbounded powers, but not of the right kind:
Consider the fixed point of $a \rightarrow aba$, $b \rightarrow ccc$,
$c \rightarrow ccc$
coded by $a,b \rightarrow 0$, $c \rightarrow 1$, generating the
word ${\bf y} = 
000111000111111111000111000 \cdots$.

As in the previous example,
$\bf y$ is uniform morphic, recurrent, and not uniformly recurrent.
If it were pure morphic, fixed by some $f$, then it would start with $f(0)^3$,
since $f(0)$ starts with $000$. But $000u000u000$ does not occur: if $k$ is
the size of the largest block of ones in $u$ (it is a power of 3),
then we have two occurrences of $01^k0$ without a larger block of ones
between them, which is always the case. In other words: our word
is the coding of $010201030102 \cdots$ 
(the ``ruler sequence'' $(\nu_2 (n))_{n \geq 1}$)
under $i \rightarrow 0001^{3^{i+1}}$.
An occurrence of $000u000u000$ could be decoded
into a square in $010201030102 \cdots$, which is square-free.
This ends our second construction for this case.
\end{example}

\begin{example}
{\bf (i) A word that is uniform primitive morphic; but not pure morphic.}

Here our example is $\bf r$, the Rudin-Shapiro sequence.    The morphism
and coding given in Example~\ref{rudin} show that $\bf r$ is uniform morphic.  
Assume it is pure morphic, generated by a morphism $g$.  Since $\bf r$ starts 
with $000$, it must be that $g({\bf r}) = {\bf r}$ contains arbitrarily large cubes 
(namely, $g^n (0) g^n(0) g^n(0)$ for all $n \geq 1$).  But from a well-known result
\cite{Allouche&Bousquet-Melou:1994b} (also see \cite{Kao&Rampersad&Shallit&Silva:2007}), 
the only cubes in $\bf r$ are $000$ and $111$, a contradiction.
\end{example}

\begin{example}
{\bf (j) A word that is pure morphic; but neither uniform morphic,
primitive morphic, nor recurrent.}

Take the Fibonacci sequence $\bf f$ discussed above in
Example~\ref{fib} and change the first symbol from $0$ to $2$.
The resulting sequence ${\bf u} = 21001010\cdots$ is pure morphic,
since it is generated by the morphism that sends
$2 \rightarrow 21$, $1 \rightarrow 0$, $0 \rightarrow 01$.  However,
it is not automatic by
Theorem~\ref{freq} since $0$ occurs with frequency $(\sqrt{5}-1)/2$,
which is irrational.  
It is neither recurrent nor primitive morphic since
$2$ occurs only once in ${\bf u}$.
\end{example}

\begin{example}
{\bf (k) A word that is pure morphic and recurrent;
but neither uniform morphic, primitive morphic, nor uniformly recurrent.}

Consider the morphism $h$ defined by $h(0) = 010$ and $h(1) = 11$.
Then $h^\omega(0)$ is evidently pure morphic.  
It is evidently recurrent because any block
that appears must appear in $h^n (0)$ for some $n$,
and then that block appears twice in $h^n(010) = h^{n+1} (0)$.

It is not uniformly recurrent because there are arbitrarily long blocks
of $1$'s.  So it is also not primitive morphic.

Suppose it is $k$-automatic for some $k$.
Using the ``logical'' approach to automatic
sequences \cite{Bruyere&Hansel&Michaux&villemaire:1994},
the sequence ${\bf u} = (u_n)_{n \geq 0}$ defined by
$$ u_n = \begin{cases}
1, & \text{if the position of the
first occurrence of a block of $n$ consecutive $1$'s in
$h^\omega(0)$} \\
& \text{is not the position of the first occurrence of a 
block of $n+1$ consecutive $1$'s}; \\
0, & \text{otherwise;}
\end{cases}
$$
is also $k$-automatic.  But this sequence (it is easy to see) is
the characteristic sequence of powers of $2$.  So we can assume $k = 2$.

Again, using the ``logical" approach,
the function $f(n)$ computing the starting
position of the first occurrence of a block of $n$ consecutive $1$'s
in the word is ``$k$-synchronized" \cite{Goc&Schaeffer&Shallit:2013},
and hence by a theorem about $k$-synchronized
sequences, we have $f(n) = O(n)$.  But it is not hard to see that in fact
$f(n) = g( \lceil \log_2 n \rceil )$,
where $g(n) = (n+2)\cdot 2^{n-1} + 1$.
So $f(n) \not\in  O(n)$, a contradiction.
\end{example}

\begin{example}  {\bf (l) A word that
is pure morphic and primitive morphic; but neither uniform morphic 
nor pure primitive morphic.}

The {\it Chacon morphism} is defined by the
map $c:  0 \rightarrow 0012, 1 \rightarrow 12, 2 \rightarrow 012$
(see \cite{Ferenczi:1995} and \cite[p.\ 133, \S 5.5.1]{Fogg:2002}).
Iterating $c$ on $0$ gives the infinite word
${\bf C} := c^\omega(0) = 0012001212012\cdots$.
The morphism $c$ is primitive,
as $c^2$ applied to each letter contains every letter.
Now consider the coding $\tau:012 \rightarrow 010$ 
Applying $\tau$ to $\bf C$ gives the
word ${\bf D} = \tau({\bf C}) = 0010001010010 \cdots$, which satisfies
${\bf D} = \delta^\omega(0)$, where $\delta:0 \rightarrow 0010, 1 \rightarrow 1$.
It follows that $\bf D$ is pure morphic and primitive morphic.

We first show that $\bf D$ is not pure primitive morphic.
Let $h$ be {\it any\/} morphism such that $h({\bf D}) = {\bf D}$.
If $h(0)=0$, we will show that $h(1)=1$. Suppose $h(1)=u$.  If
$u = \epsilon$, then $\bf D$ would be periodic, which it is not.
Otherwise, if $u$ is nonempty and not equal to $1$, then
$u$ begins with $1000$ (since ${\bf D}$ begins $001000\cdots$)
and must end with $1$ (since $u$ can be followed by $000$).
Now if $u$ ends in $001$, then since $u$ can be followed by $00u$ and 
hence by $001$ we get a contradiction,
since $001001$ is not a factor of $\bf D$.

If $u$ ends in $101$, then since $u$ can be followed by $0u$ and
hence by $01$ we get a contradiction, since $10101$
is not a factor of $\bf D$. So $u=1$.

Now let $h$ be a primitive morphism with $h({\bf D}) = {\bf D}$.
If $uu$ is a prefix of $\bf D$ then either $u=0$ or $u=t^n(0)$.
This is easy to see by induction on $|u|$. If $|u|>1$, then $u$ begins with 
$0010$,
and since every occurrence of $0010$ in $\bf D$ comes from $t(0)$,
it follows that $u =t(u')$ for some prefix $u'$ of $\bf D$. 
So $\bf D$ begins with $t(u'u')$ and hence begins with $u'u'$.

Now since $\bf D$ begins with $h(0)h(0)$ it follows that $h(0)=t^n(0)$
for some $n>0$ (if $h(0)=0$, then $h$ would not be primitive). So 
$h(0)$ begins with $0010$ and so $h(1) =t(u)$ for some $u$. Note that $u$
must contain $0$ and $1$. We can suppose that $|h(1)|$
is minimal among all
primitive morphisms $h$ fixing $\bf D$.  So
${\bf D} =t^n(0)t^n(0)t(u)t^n(0)t^n(0)t^n(0)t(u)\cdots$ so
${\bf D}=t^{n-1}(0)t^{n-1}(0)ut^{n-1}(0)t^{n-1}(0)t^{n-1}(0)u \cdots$. By
minimality of $|h(1)|$ it follows that the morphism
$h': 0 \rightarrow t^{n-1}(0)$
and $1 \rightarrow u$ is not primitive.
Since $u$ contains both $0$ and $1$ the only way
this can fail to be primitive is if $n=1$, i.e., $h(0)=0$. But now the
argument above completes the proof.

Finally, we prove that for all $k \geq 2$,
the word $\bf D$ is not $k$-automatic.
By the argument in \cite[Thm.\ 3.1]{Allouche&Allouche&Shallit:2006}, it
suffices to prove that ${\bf D} = d(0) \, d(1) \, d(2) \, \cdots $
is not $3$-automatic.  Suppose it were.  Then, by a well-known result
(e.g., \cite[Cor.\ 5.3.3]{Allouche&Shallit:2003}) the sequence 
$(d(x_n))_{n \geq 0}$ is ultimately periodic, where
$x_n = [(20)^n]_3$.  We will show it is not.

Note that $x_n = (3^{2n+1}-3)/4$.  An easy induction
shows that $|\delta^n (0)| = (3^{n+1}-1)/2$ and hence
$|\delta^{2n-1}(0) \delta^{2n-2}(0) \cdots \delta^2(0) \delta(0) 0| = 
\sum_{0 \leq i \leq 2n-1} |\delta^i(0)| = \sum_{0 \leq i \leq 2n-1} (3^{i+1}-1)/2
= (3^{2n+1}-3)/4 - n = x_n - n$.
Now another easy induction gives
$$\delta^{2n} (0) =
\delta^{2n-1}(0) \delta^{2n-2}(0) \cdots \delta(0)\, 0010 \, 1 \, \delta(0)\, 1 \, \delta^2(0)\,  1 \, \cdots 1 \,
\delta^{2n-1} (0).$$
It follows that $d(x_n)$ is the $n$'th symbol of the
infinite word ${\bf w} = 0101\delta(0) 1 \delta^2(0) 1 \delta^3(0) \cdots$.
However, $\bf w$ is not ultimately periodic.  If it were, then its subword
complexity $\rho$ would satisfy $\rho(n) \leq n$ for some $n$; however
$\bf w$ contains every prefix of $\bf D$ as a factor, and the subword
complexity of $\bf D$ is well-known to be $2n+1$, a contradiction \cite{Ferenczi:1995}.
\end{example}

\begin{example}
{\bf (m) A word that is pure primitive morphic; but not uniform morphic.}

The Fibonacci word $\bf f$ discussed above in Example~\ref{fib} is generated
by the morphism $0 \rightarrow 01;\ 1 \rightarrow 0$ and hence is pure 
primitive morphic.  But it is not automatic, as already mentioned above
in Example~\ref{fib}, and hence not uniform morphic.
\end{example}

\begin{example}
{\bf (n) A word that is pure morphic and uniform morphic, but neither pure
uniform morphic, nor primitive morphic, nor recurrent.}

Consider the morphism $h$
defined by 
\begin{eqnarray*}
3 &\rightarrow& 32 \\ 
2 &\rightarrow& 102012 \\ 
1 &\rightarrow& 1012 \\ 
0 &\rightarrow& 02 \end{eqnarray*}
An
easy induction shows that $h^{n+1} (3) = 3 g^{2n} (2) g^{2n-2}(2)
\cdots g^2 (2) 2$, where $g$ is the morphism defined in
Example~\ref{example15}.  Letting $n \rightarrow \infty$, we see that
$h^\omega(3) = 3 g^\omega(2)$.  Define ${\bf w} = h^\omega(3)$; then we
claim $\bf w$ has the desired properties.  It is clearly pure morphic,
and it is $2$-automatic because $g^\omega(2)$ is (as remarked above in
Example~\ref{example15}), and automatic sequences are closed under
shift (see, e.g., \cite[Theorem 6.8.4]{Allouche&Shallit:2003}).  In
fact, $\bf w$ is the image under the coding $\rho$ of the fixed point,
starting with $a$, of $\delta$, where $\delta(a) = ab$, $\delta(b) =
cd$, $\delta(c) = bd$, $\delta(d) = eb$, and $\delta(e) = db$, and
$\rho(abcde) = 32101$.

However, $\bf w$ is not primitive morphic because if it were,
it would be uniformly
recurrent. But $3$ only appears once, so it is not even recurrent,
a contradiction.

Finally, $\bf w$ is not pure
uniform morphic.   Suppose it is generated by a $k$-uniform morphism $f$.
If $k$ is multiplicatively independent of $2$, then $\bf w$ is both
$2$-automatic and $k$-automatic, and so by Cobham's theorem
\cite{Cobham:1969} it is ultimately periodic, a contradiction.

Therefore $k$ is multiplicatively dependent on $2$, and hence $k = 2^n$
for some $n \geq 1$.    
But now ${\bf w}[2] = 1$ and ${\bf w}[6] = 1$.  If $\bf w$ 
were the fixed point of the $2^n$-uniform morphism $f$
we would have the image of ${\bf w}[2]$ under $f$, which is
${\bf w}[2\cdot 2^n..3 \cdot 2^n-1]$, equal to the image of 
${\bf w}[6]$ under $f$, which is
${\bf w}[6 \cdot 2^n..7\cdot 2^n - 1]$.
However, from our description above we have
${\bf w} = \rho(\delta^\omega(a))$.
Since $\delta^\omega(a)$ begins with $abcdbdeb$,
it follows that ${\bf w}[2\cdot 2^n..3 \cdot 2^n-1] = \rho(\delta^n(c))$
and ${\bf w}[6 \cdot 2^n..7\cdot 2^n - 1] = \rho(\delta^n(e))$.
However $\rho(\delta^n(c))$ begins with $20$ if $n$ is odd and $10$ if
$n$ is even, whereas
$\rho(\delta^n(e))$ begins with $02$ if $n$ is odd and $12$ if $n$ is even,
a contradiction.
\label{ex20}
\end{example}

\begin{example}
{\bf (o) A word that is pure morphic and uniform morphic and recurrent,
but neither pure uniform morphic nor primitive morphic.}

Consider the fixed point of $f$: $a \rightarrow abcda$,
$b \rightarrow bcdee$, and $c,d,e \rightarrow eeeee$, followed by the
coding by $g:abcde \rightarrow 01123$.
The resulting word $\bf q$ starts
$$ 01120112333333333333011201123^{62}0 \cdots.$$
It is the fixed point of $h$ defined by
$0 \rightarrow 01120$, $1 \rightarrow 1$,
$2 \rightarrow 2333333333333$, $3 \rightarrow 33333$.
(observe that $hgf = gf^2$).
If $\bf q$ were fixed by a non-trivial uniform morphism $j$, then by 
Cobham's theorem \cite{Cobham:1969}
the length of $j$ would be $5^k$ for some $k>0$. Then we would have
$jg = gf^k$.
But $(gf^k)(b)$ starts with $1$, while $(gf^k)(c)$ starts with 
$3$, a contradiction.

The word $\bf q$ is not uniformly recurrent, because it has arbitrarily
long blocks of $3$'s.
\end{example}

\begin{example}
{\bf (p) A word that is pure morphic and uniform primitive morphic;
but neither pure uniform morphic nor pure primitive morphic.}

Let ${\bf u} = acbcbcacbcacacbc \cdots$
be the image of the Thue-Morse word $\bf t$ under the
morphism $0 \rightarrow ac$ and $1\rightarrow bc$.
Then $\bf u$ is pure morphic, because it is generated
by the morphism $a \rightarrow acb$, $b \rightarrow bca$, $c \rightarrow c$.

Now $\bf u$ is uniform primitive morphic because it is the image, 
under the coding $0123 \rightarrow acbc$ 
of the fixed point of
the word generated by the morphism $\eta$ defined
by $0 \rightarrow 01$, $1 \rightarrow 23$,
$2 \rightarrow 23$, $3 \rightarrow 01$.  Note that
$\eta^2$ applied to each letter contains every letter.

However, $\bf u$ is neither pure uniform morphic, nor pure primitive morphic.
To see this,
assume ${\bf u} = f^\omega (a)$ for some morphism $f$.
If $f$ is primitive or $k$-uniform, with $k \geq 2$, then
$f(c)$ is neither $c$ nor $\epsilon$.  Since $\bf u$ contains no 
occurrence of the factor $cc$, it must be that $f(c)$
contains an occurrence of $a$ or $b$.
Now $\bf u$ can be factored over $\{f(acbc), f(bcac) \}$, which are words
of the same length that have this $a$ or $b$ occurring
at the same position.
This implies that there is an arithmetic sequence of indices
on which $\bf u$ is constantly $a$ or constantly $b$, so the Thue-Morse
word $\bf t$ has the same property, which is not true: any sequence
extracted from Thue-Morse by indexing from an arithmetic progression
contains both $a$'s and $b$'s \cite{Gelfond:1968}.
\end{example}

\begin{example}
{\bf (q) A word that is pure primitive morphic and uniform primitive morphic;
but not pure uniform morphic.}

Let $\bf T$ be the word generated by the morphism $g$ mentioned
in Example~\ref{example15}:  $g$ maps
$2 \rightarrow 210; \ 1 \rightarrow 20;\ 0 \rightarrow 1$.
Then, as in Example~\ref{example15}, the word $\bf T$ is $2$-automatic, and the
underlying morphism is primitive, so it is uniform primitive morphic.

Suppose $\bf T$ is pure uniform morphic.  Then it is generated by
iterating a $s$-uniform morphism for some $s \geq 2$.  If $s$ is not
a power of $2$, then $\bf T$ is both $2$-automatic and $s$-automatic
where $2$ and $s$ are multiplicatively independent.  Hence by
Cobham's theorem \cite{Cobham:1969}, $T$ is ultimately periodic.
But in fact $T$ is a well-known squarefree word arising from the
Thue-Morse sequence \cite{Berstel:1979a}.  So $\bf T$ must be generated by
iterating a morphism $h$  that is $2^k$-uniform for some $k \geq 1$.
In this case, Berstel has shown that this is impossible \cite{Berstel:1979a},
because then $\bf T$ and $g({\bf T})$ differ at the position $5 \cdot 2^k$.
\end{example}

\begin{example}
{\bf (r) A word that is pure uniform morphic; but neither primitive morphic
nor recurrent.}

The word generated by the morphism $a \rightarrow ab;\ b \rightarrow bc; \ c \rightarrow cc$
iterated on $a$.  This is clearly pure uniform morphic.  
However, since $a$ only appears once, it is not recurrent and thus
cannot be primitive morphic.
\end{example}

\begin{example}
{\bf (s) A word that is pure uniform morphic and recurrent;
but not primitive morphic.}

Here we can take the word that is the fixed point of the
morphism $0 \rightarrow 010$ and $1 \rightarrow 111$, as in
\cite{Cassaigne:2001}.  This is evidently pure uniform morphic and recurrent,
but as there are arbitrarily long blocks of $1$'s, it cannot be
uniformly recurrent, and hence it is not primitive morphic.
\end{example}

\begin{example}
{\bf (t) A word that is pure uniform primitive morphic.}  The Thue-Morse word
$\bf t$.

\end{example}

\section{Final remarks}

None of the 20 examples we provided are ultimately periodic.

One might ask whether every morphic word can be generated by a coding
applied to a non-uniform morphism.  The answer is yes:
it suffices to prove this for uniform morphic
words, which is done in \cite{Allouche&Shallit:2009}.

We thank Dirk Nowotka for his helpful comments.

\newcommand{\noopsort}[1]{} \newcommand{\singleletter}[1]{#1}

\end{document}